\documentclass[11pt]{article}

\usepackage[utf8]{inputenc}
\usepackage{amsmath, amssymb, amsthm,verbatim, bbm,amsfonts, hyperref, mathrsfs}
\usepackage[margin=1.414in]{geometry}

\usepackage{amssymb,latexsym}

\usepackage{caption}
\usepackage{subcaption}
\usepackage{color}
\usepackage[pdftex]{graphicx}

\makeatletter
\newtheorem*{rep@theorem}{\rep@title}
\newcommand{\newreptheorem}[2]{%
	\newenvironment{rep#1}[1]{%
		\def\rep@title{#2 \ref{##1}}%
		\begin{rep@theorem}}%
		{\end{rep@theorem}}}
\makeatother

\newtheorem{theorem}{Theorem}

\newtheorem{definition}{Definition}
\newtheorem{lemma}{Lemma}
\newtheorem{proposition}{Proposition}
\newtheorem{corollary}{Corollary}

\newreptheorem{theorem}{Theorem}
\newreptheorem{lemma}{Lemma}
\newreptheorem{proposition}{Proposition}

\newcommand{\E}{\mathbb{E}}
\newcommand{\etal}{\emph{et al.~}}

\newcommand{\ie}{\textit{i.e.}}
\newcommand{\oo}{\mathscr{O}}
\renewcommand{\aa}{\mathscr{A}}
\newcommand{\bb}{\mathscr{B}}

\newcommand{\COMMENT}[1]{}

\newcommand{\ket}[1]{|#1\rangle}

\def\01{\{0,1\}}
\def\01{\{0,1\}}

\newcommand{\zo}{\{0,1\}}

\title{A note on the quantum query complexity of permutation symmetric functions}
\author{Andr{\'e} Chailloux\footnote{Inria de Paris, EPI SECRET, \texttt{andre.chailloux@inria.fr}}}
\date{}
\begin{document}
\maketitle

\begin{abstract}
It is known since the work of \cite{AA14} that for any permutation symmetric function $f$, the quantum query complexity is at most polynomially smaller than the classical randomized query complexity, more precisely that $R(f) = \widetilde{O}\left(Q^7(f)\right)$. In this paper, we improve this result and show that $R(f) = {O}\left(Q^3(f)\right)$ for a more general class of symmetric functions. Our proof is constructive and relies largely on the quantum hardness of distinguishing a random permutation from a random function with small range from Zhandry \cite{Zha15}.
\end{abstract}
	

\section{Introduction}

$\indent$The black box model has been a very fruitful model for understanding the possibilities and limitations of quantum algorithms. In this model, we can prove some exponential speedups for quantum algorithms, which is notoriously hard to do in standard complexity theory. Famous examples are the Deutsch-Josza problem~\cite{DJ92} and Simon's problem~\cite{Sim94}. There has been a great line of work to understand quantum query complexity, which developed some of the most advanced algorithms techniques. Even Shor's algorithm~\cite{Sho94} for factoring fundamentally relies on a black box algorithm for period finding.

We describe here the query complexity model in a nutshell. The idea is that we have to compute $f(x_1,\dots,x_n)$ where each $x_i \in [M]$ can be accessed via a query. We consider decision problems meaning that $f : S \rightarrow \zo$ with $S \subseteq [M]^n$. In this paper, we will consider inputs $x \in [M]^n$ equivalently as functions from $[n] \rightarrow [M]$. We are not interested in the running time of our algorithm but only  want to minimize the number of queries to $x$, which in the quantum setting consists of applying the unitary $\oo_x : \ket{i}\ket{j} \rightarrow \ket{i}\ket{j + x_i}$. $D(f),R(f)$ and $Q(f)$ represent the minimal amount of queries to compute $f$ with probability greater than $2/3$ (or = 1 for the case of $D(f)$) using respectively a deterministic algorithm with classical queries, a randomized algorithm with classical queries and a quantum algorithm with quantum queries. 

As we said before, the query complexity is great for designing new quantum algortihms. It is also very useful for providing black box limitations for quantum algorithms. There are some cases in particular where we can prove that the quantum query complexity of $f$ is at most polynomially smaller than classical (deterministic or randomized) query complexity. For example:
\begin{itemize}
	\item for specific functions such as search~\cite{BBBV97} or element distinctness~\cite{AS04,Kut05,Amb05}, we have respectively $Q(Search) =  \Theta(n^{1/2}), D({Search}) = \Theta(n)$ and $Q(ED) = \Theta(n^{2/3}), D(\textrm{{ED}}) =  \Theta(n)$.
	\item For any total function $f$ \ie \ when its domain $S = [M]^n$, Beals \etal~\cite{BBC+01} proved using the polynomial method that $D(f) \le O(Q^6(f))$.
\end{itemize}

Another case of interest where we can lower bound the quantum query complexity is the case of permutation symmetric functions.
There are several ways of defining such functions and we will be interested in the following definitions for a function $f : S \rightarrow \zo$ with $S \subseteq [M]^n$.
\begin{definition} $ \ $ 
\begin{itemize}
	\item $f$  permutation symmetric of the first type iff. $\forall \pi \in S_n, \ f(x) = f(x \circ \pi)$.
	\item $f$ is permutation symmetric of the second type iff. $\forall \pi \in S_n, \ \forall \sigma \in S_M,  \\ f(x) = f(\sigma \circ x \circ \pi)$.
\end{itemize}
where $S_n$ (resp. $S_M$) represents the set of permutations on $[n]$ (resp. $[M]$).

\end{definition}

\noindent Here, recall that we consider strings $x \in [M]^n$ as functions from $[n] \rightarrow [M]$. Notice also that this definition implies that $S$ is stable by permutation, meaning that $x \in S \Leftrightarrow \forall \pi \in S_n, \ x \circ \pi \in S$.  We already know from the work of Aaronson and Ambainis the following result:
\begin{theorem}[\cite{AA14}] \label{Theorem:Old}
	For any permutation invariant function $f$ of the second type, $R(f) \le \widetilde{O}(Q^7(f))$.
\end{theorem}

In a recent survey on quantum query complexity and quantum algorithms \cite{Amb17}, Ambainis writes:

\begin{quote}
``It has been conjectured since about $2000$ that a similar result also holds for $f$ with a symmetry of the first type.''
\end{quote}

\paragraph{Contribution.}
The contribution of this paper is to prove the above conjecture. We show the following:
\begin{theorem}\label{Theorem:Main}
	For any permutation invariant function $f$ of the first type, $R(f) \le O(Q^3(f))$.
\end{theorem}

This result not only generalizes the result for a more general class of permutation symmetric function, but also improves the exponent from $7$ to $3$. In the case where $M = 2$, this result was already known \cite{AA14} with an exponent of $2$, which is tight from Grover's algorithm.

 The proof technique is arguably simple, constructive and relies primarily on the quantum hardness of distinguishing a random permutation from a random function with small range from Zhandry \cite{Zha15}. We start from a permutation symmetric function $f$. At high level, the proof goes as follows:
\begin{itemize}
	\item We start from an algorithm $\aa$ that outputs $f(x)$ for all $x$ with high (constant) probability. Let $q$ the number of quantum queries to $\oo_x$ performed by $\aa$.
	\item Instead of running $\aa$ on input $x$, we choose a random function $C : [n] \rightarrow [n]$  with small range $r$ (from a distribution specified later in the paper) and apply the algorithm $\aa$ where we replace calls to $\oo_x$ with calls to $\oo_{x \circ C}$. We note that there is a simple procedure to compute $\oo_{x \circ C}$ from $\oo_x$ and $\oo_C$.
	\item If we take $r = \Theta(q^3)$, we can use Zhandry's lower bound, we show that for each $x$, the output will be close to the output of the algorithm $\aa$ where we replace calls to $\oo_{x \circ C}$ with calls to $\oo_{x \circ \pi}$ for a random permutation $\pi$. Using the fact that $f$ is permutation symmetric, the latter algorithm will output with high probability $f(x \circ \pi) = f(x)$. In other words, if the algorithm $\aa$ that calls $\oo_{x \circ C}$ wouldn't output $f(x)$ for a random $C$ and a fixed $x$ then we would find a distinguisher between a random $C$ and a random permutation $\pi$, which is hard from Zhandry's lower bound.
	\item The above tells us that applying $\aa$ where we replace calls to $\oo_x$ with calls to $\oo_{x \circ C}$ gives us output $f(x)$ with high probability. Knowing $C$, we can construct the whole string $x \circ C$ by querying $x$ on inputs $i \in Im(C)$ which can be done with $Im(C) \le r$ classical queries which allows us to construct the unitary $\oo_{x \circ C}$. This means we can emulate $\aa$ on input $x \circ C$ with $r$ classical queries to $x$ and this gives us $f(x)$ with high probability.
\end{itemize}

\noindent After presenting a few notations, we  dive directly into the proof of our theorem.

\section{Preliminaries}
\subsection{Notations}
For any function $f$ we denote by $Dom(f)$ its domain and by $Im(f)$ its range (or image).
\paragraph{Query algorithms.}
A query algorithm $\aa^{\oo}$ is described by an algorithm that calls another function $\oo$ in a black box fashion. We will never be interested in the running time or the size of $\aa$ but only in the number of calls, or queries, to $\oo$. We will consider both the cases where the algorithm $\aa^{\oo}$ is classical and quantum. In the latter $\oo$ will be a quantum unitary. In both cases, we only consider algorithms that output a single bit.
\paragraph{Oracles.}
We use oracles to perform black box queries to a function. For any function $g$, $\oo_g^{\textrm{Classical}}$ is a black box that on input $i$ outputs $g(i)$ while $\oo_g$ (without any superscript) is the quantum unitary satisfying
$$ \oo_g : \ket{i}\ket{j} \rightarrow \ket{i}\ket{j + g(i)}.$$

\paragraph{Query complexity.} Fix a function $f : S \rightarrow \zo$ where $S \subseteq [M]^n$. 

\begin{definition}
	The randomized query complexity $R(f)$ of $f$ is the smallest integer $q$ such that there exists a  classical randomized algorithm $\aa^{\oo}$ performing $q$ queries to $\oo$ satisfying: 
	$$ \forall x \in S, \ \Pr[\aa^{\oo_x^{\textrm{Classical}}} \textrm{ outputs } f(x)] \ge 2/3.$$
\end{definition}

\begin{definition}
	The quantum query complexity $Q(f)$ of $f$ is the smallest integer $q$ such that there exists a  quantum algorithm $\aa^{\oo}$ performing $q$ queries to $\oo$ satisfying: 
	$$ \forall x \in S, \ \Pr[\aa^{\oo_x} \textrm{ outputs } f(x)] \ge 2/3.$$
\end{definition}

\subsection{Hardness of distinguishing a random permutation from a random function with small range}
Our proof will use a quantum lower bound on distinguishing a random permutation from a random function with small range proven in \cite{Zha15}. Following this paper, we define, for any $r \in [n]$, the following distribution $D_r$ on functions from $[n]$ to $[n]$ from which can be sampled as follows.
\begin{itemize}
	\item Draw a random function $g$ from $[n] \rightarrow [r]$.
	\item Draw a random injective function $h$ from $[r] \rightarrow [n]$.
	\item Output the composition  $h \circ g$.
\end{itemize}

Notice that any function $f$ drawn from $D_r$ is of small range and satisfies $|Im(f)| \le r$. Let also $D_{\textrm{perm}}$ be the uniform distribution on permutations on $[n]$. Zhandry's lower bound can be stated as follows:

\begin{proposition}[\cite{Zha15}] \label{Proposition:LBCollision}
	There exists an absolute constant $\Lambda$ such that for any $r \in [n]$ and any quantum query algorithm $\mathscr{B}^\oo$ performing at most $\lceil \Lambda {r}^{1/3}\rceil$ queries to $\oo$:
	$$\forall b \in \zo, \  \left|\E_{\pi \leftarrow D_{\textrm{perm}}}\Pr[\mathscr{B^{\oo_\pi}} \textrm{ outputs } b] - \E_{C \leftarrow D_r}\Pr[\mathscr{B}^{\oo_C} \textrm{ outputs } b]\right| \le \frac{2}{27}.$$
\end{proposition}

\noindent This is obtained immediately by combining Theorem $8$ and Lemma $1$ of $\cite{Zha15}$\footnote{Equivalently, this is obtained immediately by combining Lemma 3.2 and Lemma 3.4 from the arXiv version \texttt{quant-ph:1312.1027.}}.

\section{Proving our main theorem}

The goal of this section is to prove Theorem \ref{Theorem:Main}.  Fix a function $f : S \rightarrow \zo$ where $S \subseteq [M]^n$ with $Q(f) =q$. This means there exists a quantum query algorithm $\aa^\oo$ performing $q$ queries to $\oo$ such that 
 $$ \forall x \in S, \ \Pr[\aa^{\oo_x} \textrm{ outputs } f(x)] \ge 2/3.$$
 We first amplify the success probability to $20/27$.
 \begin{lemma} \label{Lemma}
 	There exists a quantum query algorithm $\aa^\oo_3$ that performs $3q$ queries to $\oo$ such that 
 	 $$ \forall x \in S, \ \Pr[\aa^{\oo_x}_3 \textrm{ outputs } f(x)] \ge \frac{20}{27}.$$
 \end{lemma}
\begin{proof}
	$\aa^\oo_3$ will consist of the following: run $\aa^\oo$ independently $3$ times and take the output that occurs the most. For each $x$,  each run of $\aa^{\oo_x}$ outputs $f(x)$ with probability at least $2/3$. The probability that the correct $f(x)$ appears at least twice out of the $3$ results is therefore greater than $\frac{8}{27} + 3 \cdot \frac{4}{27} = \frac{20}{27}$.
\end{proof}

Using the fact that $f$ is permutation symmetric, we get the following corrolary:

\begin{corollary} \label{Corollary}
$$ \forall x \in S, \ \forall \pi \in S_n, \ 
\Pr[\aa^{\oo_{x \circ \pi}}_3 \textrm{ outputs } f(x)]
= \Pr[\aa^{\oo_{x \circ \pi}}_3 \textrm{ outputs } f(x \circ \pi)] \ge \frac{20}{27}.$$
\end{corollary}

\subsection{Looking at a small number of indices of x}

The main idea of the proof is to show that $\aa_3$ will output $f(x)$ with high probability when replacing queries to $\oo_x$ with queries to $\oo_{x \circ C}$ for $C$ chosen uniformly from $D_r$ for some $r = \Theta(Q^3(f))$.  First notice that for any $x : [n] \rightarrow [M]$ and any $g : [n] \rightarrow [n]$, it is possible to apply $\mathscr{O}_{x \circ g}$ with $2$ calls to $\oo_g$ and $1$ call to $\oo_x$ with the following procedure:
 \begin{align*}
 \ket{i}\ket{j}\ket{0} 
 \rightarrow \ket{i}\ket{j}\ket{g(i)}
 \rightarrow \ket{i}\ket{j + (x \circ g) (i)}\ket{g(i)}
 \rightarrow \ket{i}\ket{j + (x \circ g)(i)}\ket{0}
 \end{align*}
 where we respectively apply $\oo_g$ on registers $(1,3)$ ; $\oo_x$ on registers $(3,2)$ and $\oo_g^\dagger$ on registers $(1,3)$.
 
Therefore, for any fixed (and known) $x$, for any function $g : [n] \rightarrow [n]$, we can look at $\aa_3^{\oo_{x \circ g}}$ as a quantum query algorithm that queries $\oo_g$. In other words, for each $x \in S$, there is a quantum query algorithm $\bb_x^\oo$ such that $\bb_x^{\oo_g} = \aa^{\oo_{x \circ g}}$ for any function $g : [n] \rightarrow [n]$. Notice also that since a query to $\oo_{x \circ g}$ is done by doing $2$ queries to $\oo_g$, we have that $\bb^\oo$ uses twice as many queries than $\aa_3^\oo$.
 
We can now prove our main proposition that shows that we can compute $f(x)$ by looking only at $x \circ C$ meaning that we only need to look at $Im(C) \le r$ random.

\begin{proposition} \label{Proposition:Main}
	Let $f : [M]^n \rightarrow \zo$ with $Q(f) = q$ and $r = \lceil216q^3\Lambda^{-3}\rceil$ where $\Lambda$ is the absolute constant from Proposition \ref{Proposition:LBCollision}.
	$$ \forall x \in S, \ \E_{C \leftarrow D_r}\Pr[\mathscr{A}^{\oo_{x \circ C}}_3 \textrm{ outputs } f(x)] \ge 2/3.$$ 
\end{proposition}
\begin{proof}
For each $x \in S$, we consider the algorithm  $\bb_x^\oo$ described above. Recall that for all $g : [n] \rightarrow [n]$, $\bb_x^{\oo_g}= \aa_3^{\oo_{x \circ g}}$. Since $\aa_3^\oo$ uses $3q$ queries, $\bb_x^\oo$ uses $6q$ queries. We first consider the case where $g$ is a random permutation. Using Corollary \ref{Corollary}:

\begin{align*} 
\forall x \in S, \ \E_{\pi \leftarrow D_{\textrm{perm}}}\Pr[\mathscr{B}^{\oo_\pi}_x \textrm{ outputs } 0]  & = 
\E_{\pi \leftarrow D_{\textrm{perm}}}\Pr[\mathscr{A}^{\oo_{x \circ \pi}}_3 \textrm{ outputs } f(x)] \ge \frac{20}{27}
\end{align*}

\noindent Using the lower bound of Proposition \ref{Proposition:LBCollision} noticing that $6q \le \Lambda r^{1/3}$, we have 

$$ \forall x \in S, \  \left|\E_{\pi \leftarrow D_{\textrm{perm}}}\Pr[\bb^{\oo_\pi}_x \textrm{ outputs } f(x)] - \E_{C \leftarrow D_r}\Pr[\mathscr{B}^{\oo_C}_x \textrm{ outputs } f(x)]\right| \le \frac{2}{27}.$$
which gives us
$$ \forall x \in S, \ \E_{C \leftarrow D_{r}}\Pr[\mathscr{B}^{\oo_C}_x \textrm{ outputs } f(x)] \ge \frac{20}{27} - \frac{2}{27} = 2/3.$$
Since for each $x \in S$, $\mathscr{B}^{\oo_C}_x = \aa_3^{x \circ C}$, we can therefore conclude
	$$ \forall x \in S, \ \E_{C \leftarrow D_r}\Pr[\mathscr{A}^{\oo_{x \circ C}}_3 \textrm{ outputs } f(x)] \ge 2/3.$$ 
	\end{proof}

\subsection{Constructing a classical query algorithm for f}
We can now use the above proposition to prove our main theorem.

\setcounter{theorem}{1}

\begin{theorem}[Restated]
	For any permutation invariant function $f$ of the first type, $R(f) \le O(Q^3(f))$.
\end{theorem}

\begin{proof}
	Fix a function $f : S \rightarrow \zo$ where $S \subseteq [M]^n$ with $Q(f) =q$. This means there exists a quantum query algorithm $\aa^\oo$ performing $q$ queries to $\oo$ such that 
	$$ \forall x \in S, \ \Pr[\aa^{\oo_x} \textrm{ outputs } f(x)] \ge 2/3.$$
	We construct a randomized algorithm that performs $r = \lceil216q^3\Lambda^{-3}\rceil$ classical queries to $\oo_{x}^{\textrm{Classical}}$ as follows:
\begin{enumerate}
	\item Choose a random $C$ according to distribution $D_r$.
	\item Query $\oo_x^{\textrm{Classical}}$ to get all values $x_i$ for $i \in Im(C)$. This requires $|Im(C)| \le r$ queries to $\oo_x^{\textrm{Classical}}$. These queries fully characterize the function $x \circ C$, hence the quantum unitary $\oo_{x \circ C}$.
	\item From $\aa^\oo$, construct the quantum algorithm $\aa_3^\oo$ as in Lemma \ref{Lemma}. Recall that $\aa_3^\oo$ just consists of applying $\aa^\oo$ independently $3$ times and output the majority outcome.
	\item We consider $\aa_3^{\oo_{x \circ C}}$ as a quantum unitary circuit acting on $t$ qubits. At each step of the algorithm, we store the $2^{t}$ amplitudes. When $\oo_{x \circ C}$ is called, we use its representation from step $2$ to calculate its action on the $2^t$ amplitudes. Other parts of $\aa_3^{\oo_{x \circ C}}$ are treated the same way. While this uses a lot of computing power, it does not require any queries to $\oo_x^{\textrm{Classical}}$ or $\oo_x$ other than those used at step $2$.
\end{enumerate}
Step $4$ outputs the same output distribution than the quantum algorithm $\aa_3^{\oo_{x \circ C}}$. Using Proposition \ref{Proposition:Main}, for all $x \in S$, this algorithm outputs, $f(x)$ with probability greater than $2/3$, which implies 
$$R(f) \le r = \lceil216Q^3(f)\Lambda^{-3}\rceil.$$
\end{proof}

Notice that after step $2$, it is not possible to just compute $f(x \circ C)$, and try to show that it is equal to $f(x)$ since we don't even always have $x \circ C \in S$. This is yet another example in query complexity where we use the behavior of a query algorithm on inputs not necessarily in the domain of $f$.
\bibliography{paper}
\bibliographystyle{alpha}
\end{document}